\documentclass[letterpaper, 10 pt, conference]{ieeeconf}  

\IEEEoverridecommandlockouts                              
\usepackage{times}
\usepackage[utf8]{inputenc}
\usepackage{xcolor}
\usepackage{amsmath}

\usepackage{amsthm}
\usepackage{amsfonts}
\usepackage{amssymb}
\usepackage{graphicx}
\usepackage{float}
\usepackage[justification=centering,font=small,labelfont=bf]{caption}
\usepackage{tikz}
\usepackage[noadjust]{cite}
\usepackage{hyperref}
\usepackage{comment}
\usepackage{tabularx}
\allowdisplaybreaks

\DeclareMathOperator{\R}{\mathbb{R}}

\DeclareMathAlphabet{\mymathbb}{U}{BOONDOX-ds}{m}{n}

\DeclareMathOperator{\vol}{vol}

\newcommand{\lb}{\langle}
\newcommand{\rb}{\rangle}
\newcommand{\lp}{\left(}
\newcommand{\rp}{\right)}

\newcommand{\T}{\mathcal{T}}

\theoremstyle{plain}
\newtheorem{thm}{Theorem}
\newtheorem{lem}[thm]{Lemma}
\newtheorem{cor}[thm]{Corollary}
\newtheorem{prop}[thm]{Proposition}
\theoremstyle{definition}
\newtheorem{defn}[thm]{Definition}

\newtheorem{prob}[thm]{Problem}
\theoremstyle{remark}
\newtheorem*{rmk}{Remark}
\newtheorem*{sktch}{Sketch of Proof}

\newcommand\copyrighttext{%
\footnotesize\textcopyright 2023 IEEE. Personal use of this material is permitted.  Permission from IEEE must be obtained for all other uses, in any current or future media, including reprinting/republishing this material for advertising or promotional purposes, creating new collective works, for resale or redistribution to servers or lists, or reuse of any copyrighted component of this work in other works.}
\newcommand\copyrightnotice{
\begin{tikzpicture}[remember picture,overlay]
\node[anchor=south,yshift=10pt] at (current page.south) {\fbox{\parbox{\dimexpr\textwidth-\fboxsep-\fboxrule\relax}{\copyrighttext}}};
\end{tikzpicture}}

\title{\LARGE \bf
Incompressible Optimal Transport and Applications in Fluid Mixing
}

\author{Max Emerick, Bassam Bamieh
\thanks{M. Emerick and B. Bamieh are with the Department of Mechanical Engineering,
		University of California, Santa Barbara, USA
        {\tt\small \{memerick,bamieh\}@ucsb.edu}}%
}

\begin{document}

\maketitle
\thispagestyle{empty}
\pagestyle{empty}

\begin{abstract}
The problem of incompressible fluid mixing arises in numerous engineering applications and has been well-studied over the years, yet many open questions remain. This paper aims to address the question “what do efficient flow fields for mixing look like, and how do they behave?” We approach this question by developing a framework which is inspired by the dynamic and geometric approach to optimal mass transport. Specifically, we formulate the fluid mixing problem as an optimal control problem where the dynamics are given by the continuity equation together with an incompressibility constraint. We show that within this framework, the set of reachable fluid configurations can formally be endowed with the structure of an infinite-dimensional Riemannian manifold, with a metric which is induced by the control effort, and that flow fields which are maximally efficient at mixing correspond to geodesics in this Riemannian space.
\end{abstract}

\copyrightnotice

\section{Introduction}

The efficient mixing of fluids is an important problem in many engineering applications. As a result, this problem has been well-studied over the years. However, there remain many important open questions in this area.

The fluid mixing problem is usually posed as some variant of the following: given two initially unmixed fluids in some particular configuration, find an input/flow field/feedback control law which mixes the two fluids efficiently.

Early studies of fluid mixing approached this problem from a dynamical systems point-of-view. The central claim here was that efficient mixing necessitated the existence of chaotic attractors in the flow \cite{Aref2020,Ottino1989}.

Later, more quantitative studies posed the problem as one of optimization: ``find the \emph{most} efficient input/flow field/control law/etc. for mixing two fluids''. This requires one to define -- in rigorous mathematical terms -- what it means to ``mix'' two fluids, and what it means for this process to be ``efficient''. Neither of these definitions are at all obvious, and many measures of both fluid mixedness and mixing efficiency have been proposed over the years \cite{Thiffeault2012}.

Further complicating this problem is the fact that the equations of fluid mechanics -- most famously the Navier-Stokes equations, but the Euler equations as well -- are not fully understood, and depend in a delicate way on various problem parameters, including dimensionality, boundary conditions, and the regularity of allowable flow fields. Bressan \cite{Bressan2003}, for example, conjectured in 2003 that the maximum achievable rate of mixing depends on the regularity of allowable flow fields, and several results along these lines been proven since the conjecture was first posed (see e.g. \cite{Lin2011,Lunasin2012,Seis2013,Iyer2014,Zlatos2017,Alberti2019}).

On the engineering side, many methods for the control of mixing processes have been proposed as well. To mention just a few of the most relevant works, various groups have considered control for mixing by: switching between a set of prescribed velocity fields \cite{DAlessandro1999,Mathew2007,Gubanov2010,Gubanov2012}, boundary feedback control \cite{Aamo2003,Hu2018}, optimal control \cite{Liu2008,Couchman2010,Eggl2020}, gradient descent \cite{Lin2011}, and reinforcement learning \cite{Konishi2022}.

Since there is a large literature on the fluid mixing problem, it is very difficult to be exhaustive in our review, and we refer the reader to \cite{Aref2017} for further references.

One question which has not been fully answered in the literature (and to which this paper hopes to contribute something) is: ``what do efficient flow fields look like, and how do they behave?'' The motivation for asking this question is that it allows one not only to develop control strategies, but to do \emph{codesign}. In other words, the knowledge of efficient flow fields informs the design of the mixing devices themselves (not just the strategies used to control them).

In this paper, we take an abstract approach to this problem which is inspired by the dynamic/geometric approach to optimal mass transport (as in, e.g., \cite{Benamou2000,Otto2001}). In fact, the theory that we present here is essentially just ``incompressible optimal transport'' -- that is, optimal transport with an incompressibility constraint. While we are not the first to propose something like this (see, e.g., \cite{Brenier1989,Brenier1999,Liu2019,Baradat2019}), we are not aware of any other works which apply this approach to the fluid mixing problem.

The main results which come out of this approach are 1) that the set of fluid configurations which are reachable from a given starting configuration can formally be given the structure of an infinite-dimensional Riemannian manifold with a metric which is induced by the control effort, and 2) that flow fields which are maximally efficient at mixing are geodesics in this Riemannian space.

The rest of this paper proceeds as follows. In Section \ref{problem_formulation_sec}, we formulate our version of the fluid mixing problem. In Section \ref{problem_structure_sec}, we analyze its structure and prove our main results. We conclude in Section \ref{conclusion_sec} with a brief review and a discussion of future work.

\section{Problem Formulation} \label{problem_formulation_sec}

In this section, we formulate our version of the fluid mixing problem. Part of the challenge here is that there is no clear best way to do this, and thus, our first move is to try and stay away from specific assumptions and abstract as much as possible.

In the settings in which we consider, we begin with two fluids, A and B, which are contained in some domain $\Omega$. $\Omega$ is assumed to be some connected and bounded set in $\R^n$ (or more generally, some $n$-dimensional Riemannian manifold, possibly with boundary). Practically speaking, we will be most interested in the cases $n = 2$ and $n=3$, but we refrain from making any assumptions for now.

We will assume that the fluids A and B can be described by density functions, denoted by $\rho_A$ and $\rho_B$, respectively. In the classical mixing problem, $\rho_A$ and $\rho_B$ would be indicator functions which are supported on complementary sets. However, while this is an example of primary interest for us, we again refrain from making specific assumptions. For notational purposes, though, we will use $P$ to denote the space of allowable densities.

Our objective is to mix the two fluids, i.e., to bring the densities $\rho_A$ and $\rho_B$ close together. This ``closeness'' or ``mixedness'' is quantified by a metric $d$ which is defined on the space of densities $P$. Thus, $A$ and $B$ will be well-mixed if $d(\rho_A,\rho_B)$ is small, and poorly mixed if it is large. Various metrics $d$ have been proposed in the literature, including various measures of entropy, the so-called ``mix-norm'', negative-order Sobolev norms, and the Wasserstein distance from optimal mass transport \cite{Thiffeault2012}. But again, we will refrain from making specific assumptions here.

One assumption that we will make is that $\rho_A$ and $\rho_B$ are fully determined by each other (e.g., if $\rho_B = 1 - \rho_A$). Then, we will assume that instead of trying to control two densities to be close to each other, it will be equivalent to focus on just one of these densities, say, $\rho_A$, and try to control it toward some ``desired'' density $\rho_*$. In the classic mixing problem, $\rho_*$ would be the ``perfectly-mixed'' (i.e. uniform) density having the same mass as $\rho_A$. This is a mild assumption, but simplifies our problem formulation, and makes a tighter connection with the classical theory of optimal transport. Thus, in what follows, we will drop the subscript A and speak only of controlling the single density $\rho$.

Mixing typically occurs as the combination of two processes: transport and diffusion. Often, transport is primarily responsible for mixing fluids on a large scale (and is the process that we, as designers, have the most control over), while diffusion is primarily responsible for small-scale mixing and becomes significant only after large-scale mixing has already occurred. Thus, we will focus on mixing without diffusion in this paper. In this case, the dynamics of the density $\rho$ are governed by the \emph{continuity equation}
\begin{equation} \label{continuity_eqn} 
	\partial_t \rho_t(x) ~=~ - \nabla \cdot (\rho_t(x) \, v_t(x)) .
\end{equation}
Here, $\partial_t$ denotes the time derivative, $t \in [0,T]$ the time variable, $x \in \Omega$ the spatial variable, $\nabla \cdot$ the divergence operator, and $v$ the flow field/velocity field which transports the fluid. Often, the fluid is very nearly incompressible, which gives the additional \emph{incompressibility constraint}
\begin{equation} \label{incompressibility_eqn}
	\nabla \cdot v \equiv 0.
\end{equation}
It is this constraint which gives fluid mixing much of its structure as well as its associated challenges.

In addition to incompressibility, allowable flow fields $v$ will also need to satisfy some appropriate boundary conditions, typically either the \emph{no-flux} condition $v \cdot \hat{n} \equiv 0$ on $\partial \Omega$ or the \emph{no-slip} condition $v \equiv 0$ on $\partial \Omega$, depending on the model considered. Again, however, we will delay this choice, but will use $V$ to denote the space of allowable vector fields. We will assume that $V$ forms a vector space.

To define rigorously the notion of ``efficient'' mixing, we will introduce a norm $\| \cdot \|$ on the space of allowable vector fields $V$. The norm $\| \cdot \|$ is taken as a proxy for ``mixing effort'' -- a large norm meaning an aggressive flow field and a small norm meaning a conservative one. An ``efficient'' flow field is thus one that makes $d(\rho,\rho_*)$ and $\| v \|$ small simultaneously. Various norms $\| \cdot \|$ have been proposed in the literature \cite{Thiffeault2012}, including the kinetic energy $\| \cdot \|_{L^2}$, viscous dissipation energy/enstrophy $\| \nabla ( \cdot ) \|_{L^2} \approx \| \cdot \|_{H^1}$, and ``palenstrophy'' $\| \Delta ( \cdot ) \|_{L^2} \approx \| \cdot \|_{H^2}$. But again, we will refrain from making specific assumptions for now.

Putting all of this together, we formulate our version of the mixing problem as an optimal control problem which trades off between mixedness and mixing effort as follows.

\begin{prob}[Mixing]
	Given initial and target densities $\rho_i,\rho_* \in P$, time horizon $T$, norm $\| \cdot \| : V \to \R$, metric $d: P \times P \to \R$, and trade-off parameter $\alpha >0$, solve
	\begin{equation} \label{mixing_prob}
		\inf_{\rho, v} \int_0^T \| v_t \|^2 \, dt \, + \, \alpha d^2(\rho_T,\rho_*) ~~ \text{s.t.} ~ \begin{cases}
			\partial_t \rho = - \nabla \cdot (\rho v) \\
			\nabla \cdot v = 0 \\
			\rho_0 = \rho_i
		\end{cases} \hspace{-4mm}
	\end{equation}
	We will refer to $\rho$ as the state, $P$ as the state space, $v$ as the input, and $V$ as the input space.
\end{prob}

This problem generalizes the classical fluid mixing problem, since by taking $\rho_i$ to be an indicator function on some set and $\rho_*$ to be the uniform density with the same mass as $\rho_i$, we recover a classical formulation. However, the connection to optimal transport should also be clear, since by taking $\| v_t \| = \| v_t \|_{L^2(\rho_t)}$ and replacing the penalty $\alpha d^2(\rho_T,\rho_*)$ with the hard constraint $\rho_T = \rho_*$, we recover the Benamou-Brenier dynamic formulation of optimal mass transport \cite{Benamou2000}, but with an additional incompressibility constraint $\nabla \cdot v = 0$.

There is a rich structure in the above problem which we hope to effectively communicate in this paper. In short, this problem can be decomposed into a series of more fundamental geometric problems: computing geodesics, projections, critical points, and so on.

We remark that since we have not specified exactly what the spaces $P$, $V$, the norm $\| \cdot \|$, or the metric $d(\cdot,\cdot)$ are (nor how exactly we interpret solutions to \eqref{continuity_eqn}-\eqref{incompressibility_eqn}), our results are formal in nature. In particular, we will be deliberately loose about the regularity of $\rho$, $v$, and other key objects involved in our arguments, and we do not address the well-posedness of any of the stated problems.

\section{Analysis of Problem Structure} \label{problem_structure_sec}

In this section, we show that the problem \eqref{mixing_prob} can be broken down into three parts: 1) selection of a final density $\rho_f$, 2) selection of a path from $\rho_i$ to $\rho_f$, and 3) selection of a velocity field $v$ which generates that path. This is achieved by constructing a metric $m$ on the set of reachable states $[\rho_i] \subset P$ which is induced by the control effort. The final state $\rho_f$ can then be chosen as a minimizer of an optimization problem involving $m$ and $d$, the path from $\rho_i$ to $\rho_f$ as a minimizing curve in the space $([\rho_i],m)$, and the velocity field $v$ as the minimal velocity field which generates that curve. This approach is highly inspired by the Benamou-Brenier dynamic formulation of optimal mass transport \cite{Benamou2000}.

\subsection{Reachability, Metrics Induced by Control Costs, Problem Decomposition}

The first question to be addressed is that of reachability.

\begin{defn}[Reachability]
	Given two states $\rho_a, \rho_b \in P$, we say that $\rho_b$ is \emph{reachable} from $\rho_a$ on the interval $[t_i,t_f]$ if there exists a trajectory $\rho: [t_i,t_f] \to P$ and an input $v: [t_i,t_f] \to V$ such that $\rho_{t_i} = \rho_a$, $\rho_{t_f} = \rho_b$, and the pair $(\rho,v)$ satisfy the dynamics \eqref{continuity_eqn}-\eqref{incompressibility_eqn} on the interval $[t_i,t_f]$ (possibly in some suitable weak sense).
\end{defn}

The system \eqref{continuity_eqn}-\eqref{incompressibility_eqn} possesses some strong symmetries, which will lead to similarly strong conclusions about the reachability structure of the state space. The first symmetry is time-reparameterization.

\begin{lem}[Time-Reparameterization] \label{time_reparam_lem}
	Suppose that $(\rho,v)$ satisfy the dynamics \eqref{continuity_eqn}-\eqref{incompressibility_eqn} on $[t_i,t_f]$ with $\rho_{t_i} = \rho_a$ and $\rho_{t_f} = \rho_b$, and let $\sigma : [t_i',t_f'] \to [t_i,t_f]$ be any continuously differentiable and monotone nondecreasing function such that $\sigma(t_i') = t_i$ and $\sigma(t_f') = t_f$. Then the time-reparameterized functions $\tilde{\rho}(t) := \rho (\sigma(t))$, $\tilde{v}(t) := \dot{\sigma}(t) \, v(\sigma(t))$ satisfy the dynamics \eqref{continuity_eqn}-\eqref{incompressibility_eqn} on $[t_i',t_f']$, with $\tilde{\rho}_{t_i'} = \rho_a$ and $\tilde{\rho}_{t_f'} = \rho_b$.
\end{lem}

\begin{proof}
	Results directly from substituting $(\tilde{\rho},\tilde{v})$ into \eqref{continuity_eqn}-\eqref{incompressibility_eqn} and applying the chain rule.
\end{proof}

The above lemma immediately implies the following.

\begin{lem}
	If $\rho_b$ is reachable from $\rho_a$ on the finite time horizon $[t_i,t_f]$, $-\infty < t_i < t_f < \infty$, then it is reachable from $\rho_a$ on any other finite time horizon $[t_i',t_f']$.
\end{lem}

\begin{proof}
	By definition, there exists a solution $(\rho,v)$ to \eqref{continuity_eqn}-\eqref{incompressibility_eqn} with $\rho_{t_i} = \rho_a$ and $\rho_{t_f} = \rho_b$. Define $\sigma(t) := (t - t_i)^*(t_f^\prime - t_i^\prime)/(t_f-t_i) + t_i^\prime$ and $(\tilde{\rho},\tilde{v})$ as in Lemma \ref{time_reparam_lem}. Then $(\tilde{\rho},\tilde{v})$ satisfy \eqref{continuity_eqn}-\eqref{incompressibility_eqn} with $\rho_{t_i^\prime} = \rho_a$ and $\rho_{t_f^\prime} = \rho_b$.
\end{proof}

This time-reparameterization symmetry of solutions will be a key tool in many of the arguments that follow. Also, note that this justifies dropping the dependence on time when discussing reachability.

In addition to time-reparameterization symmetry, the system \eqref{continuity_eqn}-\eqref{incompressibility_eqn} also possesses time-reversibility.

\begin{lem}[Time-Reversibility] \label{time_rev_lem}
	If $(\rho,v)$ satisfy the dynamics \eqref{continuity_eqn}-\eqref{incompressibility_eqn} on $[t_i,t_f]$ with $\rho_{t_i} = \rho_a$ and $\rho_{t_f} = \rho_b$, then the time-reversed functions $\tilde{\rho}_t := \rho \big( t_f - t \big)$, $\tilde{v}_t := - v \big( t_f - t \big)$ satisfy the dynamics \eqref{continuity_eqn}-\eqref{incompressibility_eqn} on $[t_i,t_f]$, with $\tilde{\rho}_{t_i} = \rho_b$ and $\tilde{\rho}_{t_f} = \rho_a$.
\end{lem}

\begin{proof}
	Results from direct substitution.
\end{proof}

We also define the following concatenation operation which will be useful in several of our arguments.

\begin{defn}[Concatenation]
	Given pairs $(\rho_1,v_1)$ and $(\rho_2,v_2)$ satisfying the dynamics \eqref{continuity_eqn}-\eqref{incompressibility_eqn} on the intervals $[t_{1i},t_{1f}]$ and $[t_{2i},t_{2f}]$ with $\rho_1(t_{1f}) = \rho_2(t_{2i})$, we define the concatenation operation
	\begin{multline}
		\text{concat}((\rho_1,v_1),(\rho_2,v_2))(t) ~:=~ \\
		\begin{cases}
			(\rho_1(t),v_1(t)) \quad \text{if} \quad  t \in [t_{1i},t_{1f}], \\
			(\rho_2(t - t_{1f} + t_{2i}),v_2(t - t_{1f} + t_{2i})) \\ \qquad \text{if} \quad t \in (t_{1f}, t_{1f} + t_{2f} - t_{2i}]
		\end{cases} .
	\end{multline}
	It is immediate to check that $\text{concat}((\rho_1,v_1),(\rho_2,v_2))$ satisfies the dynamics \eqref{continuity_eqn}-\eqref{incompressibility_eqn} on the interval $[t_{1i},t_{1f} + t_{2f} - t_{2i}]$.
\end{defn}

These properties are all used in the proof of the following.

\begin{lem}
	Reachability is an equivalence relation on the state space $P$. That is,
	\begin{enumerate}
		\item Every state is reachable from itself,
		\item $\rho_a$ is reachable from $\rho_b$ iff $\rho_b$ is reachable from $\rho_a$,
		\item If $\rho_b$ is reachable from $\rho_a$ and $\rho_c$ is reachable from $\rho_b$, then $\rho_c$ is reachable from $\rho_a$.
	\end{enumerate}
\end{lem}

\begin{proof}
	(1) Take $v \equiv 0$ (this is allowed because $V$ is assumed to be a vector space). (2) Direct application of time-reversibility. (3) Direct application of concatenation.
\end{proof}

Recall that the equivalence relation partitions the state space $P$ into equivalence classes of mutually reachable states. In light of this fact, we will use the notation $[\rho_a]$ to denote the equivalence class of states which are reachable from $\rho_a$, and $\rho_a \sim \rho_b$ to mean that $\rho_a$ and $\rho_b$ are in the same equivalence class (i.e. are reachable from each other).

The following problem is thus foundational.

\begin{prob}[Reachability]
	Given two states $\rho_a, \rho_b \in P$, determine if $\rho_a \sim \rho_b$.
\end{prob}

Finding a general test to determine whether $\rho_a \sim \rho_b$ seems to be a surprisingly nontrivial problem. That said, we do have some simple necessary conditions. For example, any conserved quantity gives rise to a necessary condition for reachability. The most fundamental conserved quantity for the system \eqref{continuity_eqn}-\eqref{incompressibility_eqn} is the total mass of the density
\begin{equation}
	\text{mass}(\rho) ~:=~ \int_\Omega \rho(x) \, dx ,
\end{equation}
but we also have some more discerning conserved quantities.

\begin{prop}
	$\rho_\# \ell$ is a conserved quantity for the system \eqref{continuity_eqn}-\eqref{incompressibility_eqn}, where $\#$ denotes the (measure) pushforward, and $\ell$ the Lebesgue measure on $\Omega$.
\end{prop}

\begin{proof}
	By definition of the pushforward, this means that $\ell(\rho^{-1}(B)) = \vol(\rho^{-1}(B))$ is conserved for every Borel set $B \subset \R$. By the divergence theorem, the incompressibility constraint $\nabla \cdot v =0$ ensures that the volumes of all sets are conserved.
\end{proof}

\begin{cor}
	A necessary condition for $\rho_a \sim \rho_b$ is that $(\rho_a)_\# \ell = (\rho_b)_\# \ell$.
\end{cor}

\begin{proof}
	Follows directly since $\rho_\# \ell$ is conserved.
\end{proof}

The intuition for the above condition is as follows. The object $\rho_\# \ell$ is the pushforward of Lebesgue measure by the density function $\rho$ -- it is a measure on $\R$ which describes the relative volumes of the domain $\Omega$ which belong to level sets $\rho(\cdot) = r$. Under incompressibility, we cannot change the values that $\rho$ takes on, and thus we cannot change the levels themselves or the relative volumes of the level sets -- only move them around in the state space.

We also note that while this condition is necessary, it is not in general sufficient. Sufficiency seems to depend delicately on the boundary conditions, on the regularity of allowed flow fields, and on the geometry of the level sets.

Of course, knowing that two densities are mutually reachable is useless unless one can actually construct curves that connect them.
\begin{prob}[State-Transfer]
	Given initial and final states $\rho_i \sim \rho_f$, construct a input-trajectory pair $(\rho,v)$ such that $\rho_0 = \rho_i$, $\rho_1 = \rho_f$, and $(\rho,v)$ satisfy \eqref{continuity_eqn}-\eqref{incompressibility_eqn} on $[0,1]$.
\end{prob}

In general, this seems to be a similarly challenging problem to that of reachability. One idea is that perhaps by imposing more structure on the problem, it might actually become easier to solve. For example, perhaps it is easier to find optimal trajectories than it is to find arbitrary ones.

\begin{prob}[Optimal State-Transfer]
	Given initial and final states $\rho_i \sim \rho_f$ and a norm $\| \cdot \| : V \to R$, solve
	\begin{equation} \label{transfer_prob}
		\inf_{\rho, v} \, \int_0^1 \| v_t \|^2 \, dt \qquad \text{s.t.} \qquad (\star)~ \begin{cases}
			\partial_t \rho = - \nabla \cdot (\rho v) \\
			0 = \nabla \cdot v \\
			\rho_0 = \rho_i \\
			\rho_1 = \rho_f
		\end{cases}
	\end{equation}
	We will refer to the value of this problem as the \emph{transfer cost} from $\rho_i$ to $\rho_f$ and denote it as $\T(\rho_i, \rho_f)$.
\end{prob}

We will mainly be interested in the case where the norm $\| \cdot \|$ arises from an inner product $\| v_t \|^2 = \lb v_t , v_t \rb$, as this will allow us to put a Riemannian structure on the set of reachable states $[\rho_i] \subset P$. We will discuss this in more depth later. For now, it is enough to show that the above problem induces a metric on the reachable set. More precisely, we first define
\begin{equation} \label{metric_eqn}
	m(\rho_i,\rho_f) ~:=~ \inf_{\rho,v} \, \int_0^1 \| v_t \| \, dt \qquad \text{s.t.} \qquad (\star) ,
\end{equation}
where $(\star)$ denotes the constraints in \eqref{transfer_prob}. We then show that $m$ forms a metric on the reachable set $[\rho_i]$. We then show that the square root of the transfer cost $\T^{1/2}$ coincides with $m$. These properties follow directly from the symmetry properties of the system \eqref{continuity_eqn}-\eqref{incompressibility_eqn} and of the norm $\| \cdot \|$.

\begin{prop} \label{metric_prop}
	The quantity $m$ defined in \eqref{metric_eqn} forms a metric on the equivalence class of reachable states $[\rho_i]$. That is,
	\begin{enumerate}
		\item $m(\rho_a, \rho_b) \geq 0$, with equality iff $\rho_a=\rho_b$,
		\item $m(\rho_a,\rho_b) = m(\rho_b,\rho_a)$,
		\item $m(\rho_a,\rho_b) + m(\rho_b,\rho_c) \geq m(\rho_a,\rho_c)$.
	\end{enumerate}
\end{prop}

\begin{proof}
	(1) Positive definiteness follows directly from the positive definiteness of $\| \cdot \|$. To show equality, observe that $m(\rho_a,\rho_b) = 0$ implies that $\| v_t \| = 0$ for a.e-$t$ which implies that $v_t(x) = 0$ for all a.e.-$x,t$, which implies that $\partial_t \rho(x) = 0$ for all a.e.-$x,t$, and in turn, that $\rho$ is constant. Taking $v \equiv 0$ achieves equality in the converse.
	
	(2) First, we define the quantity
	\begin{equation} \label{lv_orig_eq}
		l(v) ~:=~ \int_0^1 \| v_t \| \, dt ,
	\end{equation}
	so that $m = \inf \{ \l(v) ~ \text{s.t.} ~ (\star) \}$. By time-reversibility (Lemma \ref{time_rev_lem}) and homogeneity of $\| \cdot \|$, we see that if $\rho_b$ is reachable from $\rho_a$ with cost $l(v)$ using velocity field $v_t$, then $\rho_a$ is reachable from $\rho_b$ with cost $l(v)$ using velocity field $-v_{1-t}$. The result follows by passing to the infima.
	
	(3) First, we establish that the quantity $l(v)$ defined above is invariant under time-reparameterization. That is, let $\rho , v , \sigma, \tilde{\rho}, \tilde{v}$ be defined as in Lemma \ref{time_reparam_lem}. Then it holds by change-of-variables that
	\begin{multline}
		l(\tilde{v}) ~:=~ \int_{t_i'}^{t_f'} \| \tilde{v}_t \| \, dt ~:=~ \int_{t_i'}^{t_f'} \| \dot{\sigma}(t) v(\sigma(t)) \| \, dt ~=~ \\
		\int_{t_i'}^{t_f'} \dot{\sigma}(t) \| v(\sigma(t)) \| \, dt ~=~ \int_{t_i}^{t_f} \| v_\tau \| \, d\tau ~=:~ l(v) .
	\end{multline}
	This allows us to extend the definition of $l(v)$ to arbitrary time intervals. It is then immediate that $l(v)$ adds under the concatenation operation, i.e., that $(\rho_3,v_3) = \text{concat}((\rho_1,v_1),(\rho_2,v_2))$ implies $l(v_3) = l(v_1) + l(v_2)$. The result follows by contradiction: if the inequality were violated for some $\rho_a, \rho_b, \rho_c$, then we could find velocity fields $v_1$ taking $\rho_a$ to $\rho_b$ and $v_2$ taking $\rho_b$ to $\rho_c$ with $l(v_1)$ $\epsilon$-close to $m(\rho_a,\rho_b)$ and similarly for $v_2$. Concatenation produces a $v_3$ with $l(v_3) = l(v_1) + l(v_2) < m(\rho_a,\rho_c)$, a clear contradiction since $m(\rho_a,\rho_c)$ is equal to the infimum over all possible $l(v_3)$.
\end{proof}

Now, we show that the square root of the transfer cost coincides with $m$. Actually, we show a bit more.

\begin{lem} \label{p_action_lem}
	For all $T \in (0,\infty)$, $p \in [1,\infty)$, it holds that
	\begin{equation} \label{p_action_eq}
		m^p(\rho_i,\rho_f) ~=~ \inf_{\rho,v} ~ T^{p-1} \int_0^T \| v_t \|^p \, dt \quad \text{s.t.} \quad (\star) .
	\end{equation}
\end{lem}

\begin{proof}
	Let $(\rho,v)$ be a pair satisfying $(\star)$ in the right-hand side of Equation \eqref{p_action_eq}.
	
	(1) First, we show that the pair admits a reparameterization where $\| v_t \| = l(v) / T = \text{constant}$. Define
	\begin{equation}
		\sigma(t) ~:=~ \frac{T}{l(v)} \int_0^t \| v_s \| \, ds .
	\end{equation}
	Supposing $\| v_t \|$ is continuous, $\sigma$ is continuously differentiable and monotone nondecreasing from $0$ to $T$. Therefore $\sigma$ admits a right inverse $\sigma^{-R}$, and the inverse function theorem applies in an a.e.-sense:
	\begin{equation}
		\dot{(\sigma^{-R})}(t) ~=~ \frac{1}{\dot{\sigma}(\sigma^{-R}(t))} ~=~ \frac{l(v)}{T \| v(\sigma^{-R}(t)) \|} \quad \text{a.e.}
	\end{equation}
	Define $(\tilde{\rho},\tilde{v})$ by reparamterization by $\sigma^{-R}$. Then $\tilde{v}(t) := \dot{(\sigma^{-R})}(t) \, v(\sigma^{-R}(t))$ is a velocity field on $[0,T]$, and
	\begin{multline}
		\left\| \tilde{v}_t \right\| ~=~ \left\| \dot{(\sigma^{-R})}(t) \, v(\sigma^{-R}(t)) \right\| ~=~ \\
		\left\| \frac{l(v)}{T \| v(\sigma^{-R}(t)) \|} \, v(\sigma^{-R}(t)) \right\|  ~=~ \frac{l(v)}{T} \left\| \frac{ v(\sigma^{-R}(t)) }{\| v(\sigma^{-R}(t)) \|} \right\|  \\
		~=~ \frac{l(v)}{T} ~=~ \text{constant} \quad \text{a.e.}
	\end{multline}
	
	(2) Now, we show that optimal $v$ in the right-hand side of \eqref{p_action_eq} can be taken to be of this form. This follows from Jensen's inequality, as the convexity of $\| \cdot \|^p$, $p \in [1,\infty)$ implies that
	\begin{multline} \label{jensen_eq}
		\hspace{-3mm} \int_0^T \| v_t \|^p \, dt ~=~ T \int_0^T \| v_t \|^p \, \frac{dt}{T} ~\geq~ T \lp \int_0^T \| v_t \| \, \frac{dt}{T}  \rp^p \\
		~=~ T \left( \frac{l(v)}{T} \right)^p ~=~ T^{1-p} \, l^p(v) ,
	\end{multline}
	with equality if $\| v_t \| = l(v) / T = \text{constant}$.
	
	(3) The result follows from scaling Equation \eqref{jensen_eq} by $T^{p-1}$ and passing to the infimum, recognizing that $m = \inf \{ \l(v) ~ \text{s.t.} ~ (\star) \}$ (see proof of Lemma \ref{metric_prop}).
\end{proof}

In particular, taking $T = 1$ and $p = 2$ above gives $m = \T^{1/2}$. In other words, the metric $m$ can be characterized as the square root of the transfer cost.

Now, observe that the transfer cost of the optimal state-transfer problem \eqref{transfer_prob} is exactly the same as the control energy portion of the cost for the mixing problem \eqref{mixing_prob} when $\rho_f = \rho_T$ and $T = 1$. This allows us to break the infimization in the mixing problem over $\rho_T = \rho_1$ as follows
\begin{align}
	&\inf_{\rho, v} \int_0^T \| v_t \|^2 \, dt  +  \alpha d^2(\rho_T,\rho_*) ~~ \text{s.t.} ~ \begin{cases}
		\partial_t \rho = - \nabla \cdot (\rho v) \\
		0 = \nabla \cdot v \\
		\rho_0 = \rho_i
	\end{cases}  \hspace{-7mm} \\
	&=~ \inf_{\rho_f}\, \frac{1}{T} \Bigg[ \inf_{\rho, v} \, \int_0^1 \| v_t \|^2 \, dt ~~ \text{s.t.} ~~ (\star)  \Bigg]  \, \nonumber \\
	& \qquad \qquad \qquad \qquad + \,  \alpha d^2(\rho_f,\rho_*) ~~ \text{s.t.} ~~ \rho_f \in [\rho_i] \\
	&=~ \inf_{\rho_f}\, \frac{1}{T} \, m^2(\rho_i,\rho_f)  \, + \, \alpha d^2(\rho_f,\rho_*) ~~ \text{s.t.} ~~ \rho_f \in [\rho_i] . \label{state_selection_eq}
\end{align}
Note that the factor $1/T$ above comes from rescaling the problem into the time interval $[0,1]$. With the above argument, we have just proven the following.

\begin{thm}
	Any optimal solution to the mixing problem \eqref{mixing_prob} is given by a (linearly time-scaled) optimal solution to the state-transfer problem \eqref{transfer_prob} where $\rho_f$ is determined from the solution to the final state selection problem \eqref{state_selection_eq}.
\end{thm}

\begin{proof}
	Follows directly as outlined above.
\end{proof}

We have now broken the mixing problem down into two component subproblems: selection of a final state $\rho_f$ according to \eqref{state_selection_eq}, and the optimal state-transfer problem \eqref{transfer_prob} from $\rho_i$ to $\rho_f$. In principle, a suitably well-mixed final state $\rho_f$ can be found by solving \eqref{state_selection_eq} for a sufficiently large value of $\alpha$. In practice, this may be computationally infeasible. We will have more to say about this in follow-up work.

Now, we focus on solving the optimal state-transfer problem \eqref{transfer_prob}. We will show that this problem can be further broken down into selection of a path in the space of reachable densities $\rho: [0,1] \to [\rho_i] \subset P$ and selection of a velocity field $v$ which generates that path. The key step here is to use the \emph{metric derivative}.

\begin{defn}[Metric Derivative]
	Consider the space of reachable densities $[\rho_i] \subset P$ endowed with the metric $m$. Given an absolutely continuous curve $\rho:[0,1] \to ([\rho_i],m)$, we define its \emph{metric derivative}
	\begin{equation} \label{metric_deriv_def}
		|\dot{\rho}_t | ~:=~ \lim_{h\to 0} \frac{m(\rho_{t+h},\rho_t)}{|h|}  .
	\end{equation}
	The definition of absolute continuity ensures that the metric derivative is well-defined for a.e. $t \in [0,1]$.
\end{defn}

The metric derivative can be related back to the velocity field $v$ as follows.

\begin{lem}
	Let $\rho: [0,1] \to ([\rho_i],m)$ be absolutely continuous. Then its metric derivative satisfies for a.e. $t \in [0,1]$
	\begin{equation} \label{metric_derivative_lem}
		|\dot{\rho}_t| ~=~ \inf_{v_t} \| v_t \|  \quad \text{s.t.} \quad (*) ~ \begin{cases}
			\partial_t \rho_t = - \nabla \cdot (\rho_t v_t) \\
			0 = \nabla \cdot v_t
		\end{cases} 
	\end{equation}
\end{lem}

\begin{proof}
	The main idea of the proof is to recognize that the metric $m$ is induced by the following length structure:
	\begin{align}
		L(\rho) ~&:=~ \inf_v \, l(v) ~~ \text{s.t.} ~~ (*) ~ \forall t \\
		~&:=~ \inf_v \int_0^1 \| v_t \| \, dt ~~ \text{s.t.} ~~ (*) ~ \forall t \\
		~& \phantom{:} =~ \int_0^1 \left[ \inf_{v_t} \| v_t \| ~~ \text{s.t.} ~~ (*) \right] \, dt . \label{inf_v_eq}
	\end{align}
	That is, that
	\begin{align}
		m(\rho_a , \rho_b) ~&:=~ \inf_{\rho,v} \, \int_0^1 \| v_t \| \, dt ~~ \text{s.t.} ~~ (\star) \\
		~& \phantom{:} =~ \inf_\rho L(\rho) ~~ \text{s.t.} ~~ \rho_0 = \rho_a , ~ \rho_1 = \rho_b .
	\end{align}
	The fact that $L$ is a length structure follows from the fact that $l$ is a length structure, which follows from the properties of $l$ established in the proof of Proposition \ref{metric_prop}. The equality in \eqref{inf_v_eq} follows from the fact that $v$ can be infimized pointwise-in-time, which follows from the fact that the constraints $(*)$ are pointwise-in-time.
	
	With the metric $m$ thus defined, we can then define the \emph{intrinsic length} of an absolutely continuous curve $\rho$ by
	\begin{equation}
		L_m(\rho) ~:=~ \int_0^1 | \dot{\rho}_t | \, dt .
	\end{equation}
	It is a foundational result in metric geometry (see, e.g., \cite[Theorem 2.4.3]{Burago2001} and surrounding discussion) that if $L$ is lower-semicontinuous with respect to the uniform convergence on $\rho$, then the intrinsic length $L_m$ coincides with the length $L$ as defined via the length structure. Supposing this to be true, then
	\begin{equation}
		\int_0^1 | \dot{\rho}_t | \, dt  ~=~ \int_0^1 \left[ \inf_{v_t} \| v_t \| ~~ \text{s.t.} ~~ (*) \right] \, dt
	\end{equation}
	for all absolutely continuous curves $\rho$, so that by the additivity of $L$,$L_m$ under concatenation, we must have
	\begin{equation}
		|\dot{\rho}_t| ~=~ \inf_{v_t} \| v_t \|  \quad \text{s.t.} \quad (*)
	\end{equation}
	for a.e. $t \in [0,1]$.
\end{proof}

Note that in equations \eqref{metric_deriv_def} and \eqref{metric_derivative_lem}, the curve $\rho$ is considered to be fixed. In particular, $(*)$ is considered as a set of static constraints for $v$ at each $t$ (depending on $\rho$) rather than as a dynamic constraint. We can then break the optimal state-transfer problem \eqref{transfer_prob} over paths $\rho$ as follows.
\begin{align}
	&\inf_{\rho, v} \, \int_0^1 \| v_t \|^2 \, dt \qquad \text{s.t.} \qquad (\star) \\
	&=~ \inf_{\rho}  \int_0^1 \left[ \inf_{v_t} \| v_t \|^2 ~~ \text{s.t.} ~~ (*) \right] \, dt
	~~ \text{s.t.} ~ \begin{cases}
		\rho \in [\rho_i] ~ \text{cts} \\
		\rho_0 = \rho_i \\
		\rho_1 = \rho_f
	\end{cases} \hspace{-7mm} \\
	&=~ \inf_{\rho}  \int_0^1 | \dot{\rho_t} |^2 \, dt
	~~ \text{s.t.} ~~ \begin{cases}
		\rho \in [\rho_i] ~ \text{cts} \\
		\rho_0 = \rho_i \\
		\rho_1 = \rho_f
	\end{cases} .
\end{align}

Recall from the proof of Lemma \ref{p_action_lem} that optimal paths must have constant $\| v_t \|$, thus constant $|\dot{\rho}|$, and thus
\begin{equation}
	\int_0^1 | \dot{\rho}_t |^2 \, dt ~=~ \left( \int_0^1 |\dot{\rho}_t| \, dt \right)^2 ~=~ \text{length}^2(\rho) .
\end{equation}

We have therefore proven the following.

\begin{thm}
	Any optimal solution to the state-transfer problem \eqref{transfer_prob} is given by a pair $(\rho,v)$, where $\rho$ is a minimum-length continuous path in the reachable set $[\rho_i] \subset P$ with endpoints $\rho_i$ and $\rho_f$ and speed $|\dot{\rho}| \equiv m(\rho_i,\rho_f) = \text{constant}$, and $v$ minimizes at each time $t$ the left-hand side of \eqref{metric_derivative_lem}.
\end{thm}

\begin{proof}
	Follows directly as outlined above.
\end{proof}

\begin{rmk}
	The minimum-length curves connecting each $\rho_i$ to a given $\rho_f$ can also be characterized as steepest-descent curves (i.e. gradient flows) of the distance function $m(\cdot,\rho_f)$. If such gradients can be defined uniquely, this provides a feedback controller $k: [\rho_i] \to V$ achieving optimal state transfer. This is classical, and connects the metric space picture with optimal control theory and the HJB equation.
\end{rmk}

We have now effectively decomposed our problem into three component subproblems: 1) selection of a final state $\rho_f$, 2) selection of a minimum-length continuous path which connects $\rho_i$ to $\rho_f$, and 3) selection of a minimum-norm velocity field $v$ which generates that path.

At this point, however, we still have not said anything constructive about this problem. Given the metric $m$, we could compute minimum-length paths by gradient descent, and similarly, if we knew the minimum-length paths, then we could integrate to find $m$. However, at the moment, we have neither. The next subsection takes the first steps towards more constructive solutions.

\subsection{Riemannian Structure, Necessary Conditions, Geodesics}

Here, we describe how the reachable space $([\rho_i],m)$ can formally be viewed as an infinite-dimensional Riemannian manifold. We also present necessary conditions for optimality for the problems posed in the previous section, and give interpretations in terms of this Riemannian framework. This section is highly inspired by Otto's work on the geometry of the 2-Wasserstein space from optimal transport \cite{Otto2001}.

To describe a Riemannian manifold, we need three things:
\begin{enumerate}
	\item a smooth manifold $M$,
	\item a description of the tangent spaces $T_x M$, 
	\item an inner product $\lb \cdot , \cdot \rb_x$ on each tangent space.
\end{enumerate}
In our setting, we identify each of these by the following:
\begin{enumerate}
	\item the set of reachable states $[\rho_i] \subset P$,
	\item the set of perturbations $\tau$ at a given $\rho$ satisfying
	\begin{equation}
		\tau = - \nabla \cdot (\rho v) \quad \text{for some} \quad v \quad \text{s.t.} \quad \nabla \cdot v = 0 ,
	\end{equation}
	\item the inner product
	\begin{equation} \label{inner_prod_eq}
		\hspace{-2mm}
		\lb \tau_1 , \tau_2 \rb_\rho ~:=~ \inf_{v_1,v_2} \,  \lb v_1 , v_2 \rb_V \quad \text{s.t.} \quad 
		\begin{cases}
			\tau_1 = - \nabla \cdot (\rho v_1) \\
			\tau_2 = - \nabla \cdot (\rho v_2) \\
			0 = \nabla \cdot v_1 \\
			0 = \nabla \cdot v_2
		\end{cases}
	\end{equation}
\end{enumerate}
In the above, $\lb \cdot , \cdot \rb_V$ is the inner product on $V$ that was assumed to induce the norm $\| \cdot \|$ in the previous section. We can see that this indeed recovers the metric as defined earlier, since by taking $\dot{\rho} = \partial_t \rho = \tau_1 = \tau_2$ in \eqref{inner_prod_eq}, we obtain exactly the characterization \eqref{metric_derivative_lem} of the metric derivative.

At this point, we could take a route parallel to that of Otto \cite{Otto2001}, and show that this Riemannian structure is induced by submersion on Arnold's group of volume-preserving diffeomorphisms on $\Omega$ \cite{Arnold1966}
\begin{equation}
	\text{Sdiff}(\Omega) ~:=~ \{ \phi \in \text{diff}(\Omega) ~:~  \det D \phi \equiv 1 \} ,
\end{equation}
with an appropriate metric inherited from $\lb \cdot , \cdot \rb_V$. We will say more about this in follow-up work. Here, our main objective is to build towards some constructive results concerning flow fields for efficient mixing, and to interpret these results in terms of this Riemannian structure.

First, we examine necessary conditions for optimality for the \emph{velocity field selection problem}, restated as follows:
\begin{equation} \label{vel_selection_eqn}
	\inf_{v} \| v \|^2 \qquad \text{s.t.} \qquad \begin{cases}
		\tau = - \nabla \cdot (\rho v) \\
		0 = \nabla \cdot v
	\end{cases} .
\end{equation}
Recall that $\tau$ and $\rho$ are assumed to be given in the above problem.
We also assume that the norm $\| \cdot \|$ is given by
\begin{align} \label{norm_K_eqn}
	\| \cdot \|_V^2 ~=~ \| K (\cdot) \|_{L^2} ~&=~ \langle K(\cdot),K(\cdot) \rangle_{L^2} \nonumber \\
	~&=~  \langle K^*K(\cdot),(\cdot) \rangle_{V \hookrightarrow L^2}
\end{align}
for some linear operator $K: V \to L^2$. (This includes, for example, the $L^2$ norm ($K=I$), $H^1$ norm ($K = (I - \Delta)^{1/2}$), enstrophy norm ($K = \nabla$), and palenstrophy norm ($K = \Delta$).)

\begin{thm} \label{vel_sel_necc_cond_thm}
	For $v$ to be optimal in \eqref{vel_selection_eqn} under assumption \eqref{norm_K_eqn}, it is necessary that there exist $\lambda,\gamma : \Omega \to \R$ such that
	\begin{align}
		\tau &~=~ - \nabla \cdot (\rho v) \\
		0 &~=~ \nabla \cdot v \\
		0 &~=~ K^*K v - \rho \nabla \lambda - \nabla \gamma .
	\end{align}
\end{thm}

\begin{proof}
	We solve this problem by the method of Lagrange multipliers. We first form the Lagrangian
	\begin{equation}
		L(v,\lambda,\gamma) \, := \, \frac{1}{2} \langle K v , K v \rangle + \langle \lambda , \tau + \nabla \cdot (\rho v) \rangle + \langle \gamma , \nabla \cdot v \rangle .
	\end{equation}
	Note that we have scaled the objective in \eqref{vel_selection_eqn} by a factor of $1/2$ -- this simplifies the analysis but does not alter the resulting optimality conditions. We have also substituted the form \eqref{norm_K_eqn} for $\| v \|^2$. We now take variations with respect to each parameter, consider the linear parts, and set them equal to zero. This yields the equations
	\begin{align}
		0 ~&=~ \lb \tilde{\lambda} , \tau + \nabla \cdot (\rho v) \rb \\
		0 ~&=~ \lb \tilde{\gamma} , \nabla \cdot v \rb \\
		0 ~&=~ \lb K^*K v , \tilde{v} \rb + \lb \lambda , \nabla \cdot (\rho \tilde{v}) \rb + \lb \gamma , \nabla \cdot \tilde{v} \rb .
	\end{align}
	Since these equations must hold for all perturbations $\tilde{\lambda}, \tilde{\gamma}, \tilde{v}$, we can see that the right hand parts of the inner products in the first two equations must be equal to zero, and we recover our original constraints: $\tau = -\nabla \cdot (\rho v)$ and $\nabla \cdot v = 0$. In the third equation, we first use the divergence theorem to integrate by parts and then combine terms to obtain
	\begin{equation}
		0 ~=~ \lb K^*K v - \rho \nabla \lambda - \nabla \gamma , \tilde{v} \rb .
	\end{equation}
	By the same reasoning, since this equation must hold for all perturbations $\tilde{v}$, the left part of the inner product must be equal to zero, and thus we obtain the system
	\begin{align}
		\tau &~=~ - \nabla \cdot (\rho v) \\
		0 &~=~ \nabla \cdot v \\
		0 &~=~ K^*K v - \rho \nabla \lambda - \nabla \gamma .
	\end{align}
\end{proof}

In the context of the Riemannian submersion, these are the necessary conditions for a vector field $v$ with image $\tau$ to be \emph{horizontal} with respect to the submersion on $\text{Sdiff}(\Omega)$. We will say more about this in follow-up work.

We also remark that in many cases, we expect $K^*K$ to be invertible so that $v$ can be chosen as
\begin{equation}
	v ~=~ (K^*K)^{-1}(\rho \nabla \lambda + \nabla \gamma) .
\end{equation}
We will say more about this equation and its interpretation in follow-up work.

Now, we examine necessary conditions for optimality for the \emph{optimal state-transfer problem} \eqref{transfer_prob}.

\begin{thm} \label{state_transfer_necc_conds_thm}
	For the pair $(\rho,v)$ to be optimal in \eqref{transfer_prob} under assumption \eqref{norm_K_eqn}, it is necessary that there exist $\lambda,\gamma : [0,T] \times \Omega \to \R$ such that
	\begin{align}
		\partial_t \rho &~=~ - \nabla \cdot (\rho v) \label{necc_cond_dyn_eqn} \\
		0 &~=~ \nabla \cdot v \\
		0 &~=~ \partial_t \lambda + \nabla \lambda \cdot v \\
		0 &~=~ K^*K v - \rho \nabla \lambda - \nabla \gamma  , \label{necc_cond_v_eqn}
	\end{align}
	in addition to the boundary conditions $\rho_0 = \rho_i$, $\rho_1 = \rho_f$.
\end{thm}

\begin{proof}
	We solve this problem by the method of Langrange multipliers just as in the proof of Theorem \ref{vel_sel_necc_cond_thm}. The Lagrangian is given by
	\begin{equation}
		L(\rho,v,\lambda,\gamma) \, := \, \frac{1}{2} \langle K v , K v \rangle + \langle \lambda , \partial_t \rho + \nabla \cdot (\rho v) \rangle + \langle \gamma , \nabla \cdot v \rangle .
	\end{equation}
	After taking variations and setting the linear parts equal to zero, we obtain
	\begin{align}
		0 ~&=~ \lb \tilde{\lambda} , \partial_t \rho + \nabla \cdot (\rho v) \rb \\
		0 ~&=~ \lb \tilde{\gamma} , \nabla \cdot v \rb \\
		0 ~&=~ \lb \lambda , \partial_t \tilde{\rho} + \nabla \cdot (\tilde{\rho} v) \rb \\
		0 ~&=~ \lb K^*K v , \tilde{v} \rb + \lb \lambda , \nabla \cdot (\rho \tilde{v}) \rb + \lb \gamma , \nabla \cdot \tilde{v} \rb .
	\end{align}
	The first, second, and fourth equations are identical to those investigated in the proof of Theorem \ref{vel_sel_necc_cond_thm} (with $\tau = \partial_t \rho$), and we obtain the same equations. For the third equation, we first split the inner product into terms involving $\partial_t \tilde{\rho}$ and $\nabla \cdot (\tilde{\rho} v)$, integrate the first in time and the second in space (again using the divergence theorem) to obtain
	\begin{equation}
		0 ~=~ \lb - \partial_t \lambda , \tilde{\rho} \rb + \lb -\nabla \lambda , \tilde{\rho} v \rb .
	\end{equation}
	Finally, we move $v$ into the left side of the second inner product and combine terms to obtain
	\begin{equation}
		0 ~=~ -\lb \partial_t \lambda + \nabla \lambda \cdot v , \tilde{\rho} \rb .
	\end{equation}
	By the same reasoning as before, the left side of the inner product must be equal to zero. Thus we obtain the system 
	\begin{align}
		\partial_t \rho &~=~ - \nabla \cdot (\rho v) \label{necc_cond_dyn_eqn} \\
		0 &~=~ \nabla \cdot v \\
		0 &~=~ \partial_t \lambda + \nabla \lambda \cdot v \\
		0 &~=~ K^*K v - \rho \nabla \lambda - \nabla \gamma  . \label{necc_cond_v_eqn}
	\end{align}
	Since the curve $\rho$ has fixed endpoint constraints, this system has the transversality conditions $\rho_0 = \rho_i$, $\rho_1 = \rho_f$.
\end{proof}

The above conditions are exactly the original dynamics (including incompressibility), the condition that $v_t$ be optimal for each $\partial_t \rho_t$ (as in Theorem \ref{vel_sel_necc_cond_thm}), and a new \emph{costate equation}, describing the evolution of the costate $\lambda$ that enforces optimality. Since optimal here means shortest path, the interpretation of the costate equation in the Riemannian context is as a geodesic equation. (In Riemannian geometry, the geodesic equation describes curves of zero intrinsic acceleration, and can also be derived via the Riemannian curvature. Here, these curves are derived via a variational principle instead.)
Thus, on one hand, these necessary conditions describe geodesics/shortest paths in this Riemannian framework. On the other hand, since we have constructed the metric $m$ from the transfer cost in the optimal state-transfer problem, these geodesics exactly describe flow fields which transfer states with minimal effort (i.e. with maximal efficiency). Therefore, maximally efficient mixing fields must always be geodesics of this form.

\section{Conclusion} \label{conclusion_sec}

In this paper, we formulated a version of the fluid mixing problem as an optimal control problem which trades off between fluid mixedness and mixing effort. We showed that within this framework, the set of reachable fluid configurations can formally be endowed with a Riemannian metric induced by the control effort, and that the mixing problem can be decomposed into a series of more fundamental geometric problems. In particular, we showed that flow fields which are maximally efficient at mixing must be geodesics in this Riemannian space.

In terms of future work, there is still much to be done. First, while the geometry of incompressible flows has been thoroughly investigated in the flow-map setting (i.e. in $\text{Sdiff}(\Omega)$) with the $L^2$ geometry, the connection with the density setting, optimal transport, and fluid mixing remain relatively unexplored. Second, while the $L^2$ case has rich structure and many connections with existing theories, it is probably not the best penalty from an engineering codesign point-of-view, and other penalizations should be explored as well. Third, we have not said anything yet about the final-state selection problem or its connection with existing theories of fluid mixing. These are all potentially fruitful directions for future work.





\bibliographystyle{ieeetr}
\bibliography{library}

~~~

~~~

\end{document}